\documentclass[11pt]{amsproc}

\usepackage{amsmath,amssymb,amsthm}
\usepackage{fullpage}
\usepackage{amsfonts}
\usepackage{latexsym}
\usepackage{amstext}
\usepackage{cite}
\usepackage{graphicx}
\usepackage{verbatim}
\usepackage{comment}
\usepackage{soul,color}
\usepackage{xspace} 

\newtheorem{theorem}{Theorem}[section]
\newtheorem{lemma}[theorem]{Lemma}
\newtheorem{corollary}[theorem]{Corollary}
\newtheorem{conjecture}[theorem]{Conjecture}

\newtheorem{definition}[theorem]{Definition}

\def\BT{\begin{theorem}}
\def\ET{\end{theorem}}
\def\BL{\begin{lemma}}
\def\EL{\end{lemma}}
\def\BC{\begin{corollary}}
\def\EC{\end{corollary}}
\def\BE{\begin{example}}
\def\EE{\end{example}}
\def\BD{\begin{definition}}
\def\ED{\end{definition}}
\def\BR{\begin{remark}}
\def\ER{\end{remark}}
\def\BAS{\begin{assumption}}
\def\EAS{\end{assumption}}
\def\BCJ{\begin{conjecture}}
\def\ECJ{\end{conjecture}}
\def\BPF{\begin{proof}}
\def\EPF{\end{proof}}
\def\beq{\begin{equation}}
\def\eeq{\end{equation}}
\def\R{{\bf R \,}}
\def\rfont{\tt}
\def\pfont{\sc}

\def\LOCAL{$\mathcal{LOCAL}$\xspace}
\def\GOSSIP{$\mathcal{GOSSIP}$\xspace}
\def\CONGEST{$\mathcal{CONGEST}$\xspace}
\def\Rumor{{\pfont Rumor}\xspace}
\def\NeighborExchange{{\pfont NeighborExchange}\xspace}

\newcommand{\eps}{\epsilon}

\newcommand{\polylog}{\mathrm{polylog}}
\newcommand{\argmin}{\mathrm{argmin}}

\DeclareMathOperator{\vol}{vol}

\begin{document}

\setcounter{page}{1}

\title{Global Computation in a Poorly Connected World:\\ Fast Rumor Spreading with
No Dependence on Conductance}

%
\thanks{This work was partially supported by the Simons Postdoctoral Fellows Program and NSF grant CCF-0843915.}
\maketitle
\begin{center}
\parbox{1.4in}{\center Keren Censor-Hillel\\MIT}
\parbox{1.4in}{\center Bernhard Haeupler\\MIT}
\parbox{1.4in}{\center Jonathan A.\ Kelner\\MIT}
\parbox{1.4in}{\center Petar Maymounkov\\MIT}
\end{center}
{\ }\\{\ }\\
\date{}

\thispagestyle{empty}

\begin{abstract}
In this paper, we study the question of how efficiently a collection of interconnected nodes can perform a global computation in the widely studied \GOSSIP model of communication.  In this model, nodes do not know the global topology of the network, and they may only initiate contact with a single neighbor in each round.  This model contrasts with the much less restrictive \LOCAL model, where a node may simultaneously communicate with all of its neighbors in a single round.  A basic question in this setting is how many rounds of communication are required for the information dissemination problem, in which each node has some piece of information and is required to collect all others.

In the \LOCAL model, this is quite simple: each node broadcasts all of its information in each
round, and the number of rounds required will be equal to the diameter of the underlying communication graph.
In the \GOSSIP model, each node must independently choose a single neighbor to contact, and the lack of global information makes it difficult to make any sort of principled choice.  As such, researchers have focused on the \emph{uniform gossip algorithm}, in which each node independently selects a neighbor uniformly at random.  When the graph is well-connected, this works quite well.  In a string of beautiful papers, researchers proved a sequence of successively stronger bounds on the number of rounds required in terms of the conductance $\phi$, culminating in a bound of $O(\phi^{-1} \log n)$.

In this paper, we show that a fairly simple modification of the protocol gives an algorithm that solves the information dissemination problem in at most $O(D+\text{polylog}{(n)})$ rounds in a network of diameter $D$, with \emph{no dependence on the conductance}. This is at most an additive polylogarithmic factor from the trivial lower bound of $D$, which applies even in the \LOCAL model.

In fact, we prove that something stronger is true: \emph{any} algorithm that requires $T$ rounds in the \LOCAL model can be simulated in $O(T +\mathrm{polylog}(n))$ rounds in the \GOSSIP model.
We thus prove that these two models of distributed computation are essentially equivalent.

\end{abstract}

\newpage
\setcounter{page}{1}
\section{Introduction}

Many distributed applications require nodes of a network to perform a global task using only local knowledge. Typically a node initially only knows the identity of its neighbors and gets to know a wider local neighborhood in the underlying communication graph by repeatedly communicating with its neighbors. Among the most important questions in distributed computing is how certain global computation problems, e.g., computing a maximal independent set~\cite{luby} or a graph coloring~\cite{BarenboimE2009}, can be performed with such local constraints.
\begin{sloppypar}Many upper and lower bounds for distributed tasks are given for the well-known \LOCAL model~\cite[Chapter 2]{Peleg2000}, which operates in synchronized rounds and allows each node in each round to exchange messages of unbounded size with all of its neighbors. It is fair to say that the \LOCAL model is essentially the established minimal requirement for a distributed algorithm. Indeed, whenever a distributed algorithm is said to have running time $T$ it is implied that, at the least, there exists a $T$-round algorithm in the \LOCAL model.
\end{sloppypar}

In many settings, practical system design or physical constraints do not allow a node to contact all of its (potentially very large number of) neighbors at once. In this paper we focus on this case and consider the \GOSSIP model, which restricts each node to initiate at most one (bidirectional) communication with one of its neighbors per round. In contrast to computations in the \LOCAL model, algorithms for the \GOSSIP model have to decide which neighbor to contact in each round. This is particularly challenging when the network topology is unknown.

Algorithms with such \emph{gossip constraints} have been intensively
studied for the so-called \Rumor problem (also known as the \emph{rumor spreading} or \emph{information dissemination} problem), in which each node has some initial input and is required to collect the information of all other nodes. Most previous papers analyzed the simple {\rfont UniformGossip} algorithm, which chooses a random neighbor to contact in each round. The uniform gossip mixes well on well-connected graphs, 
 and good bounds for its convergence in terms of the graph conductance have been given~\cite{ChierichettiLP2010,MoskAoyamaS2006a,Giakkoupis2011}. However, it has a tendency to repeatedly communicate between well-connected neighbors while not transmitting information across bottlenecks.
  Only recently have algorithms been designed that try to avoid this behavior.
By alternating between random and deterministic choices,~\cite{CensorHillelShachnai2011} showed that fast convergence can be achieved for a wider family of graphs, namely, those which have large \emph{weak conductance} (a notion defined therein).  However, while this outperformed existing techniques in many cases, its running time bound still depended  on a notion of the connectivity of the graph.

\subsection{Our results}
This paper significantly improves upon previous algorithms by  providing
 the first information spreading algorithm for the \GOSSIP model
that is fast for \emph{all} graphs, with no dependence on their conductance.
Our algorithm requires at most $O(D+\text{polylog}{(n)})$ rounds in a network of size $n$ and diameter $D$. This is at most an additive polylogarithmic factor from the trivial lower bound of $O(D)$ rounds even for the \LOCAL model. In contrast, there are many graphs with polylogarithmic diameter on which all prior algorithms have $\Omega(n)$ bounds.

In addition, our results apply more generally to any algorithm in the \LOCAL
model. We show how any algorithm that takes $T$ time in the \LOCAL
model can be simulated in the \GOSSIP model in $O(T + \text{polylog}{(n)})$ time, thus incurring only an additional 
polylogarithmic cost in the size of the network $n$. Our main result
that leads to this simulation is an algorithm for the \GOSSIP model in
which each node exchanges information (perhaps indirectly) with each
of its neighbors within a polylogarithmic number of rounds. This holds
for every graph, despite the possibility of large degrees. A key
ingredient in this algorithm is a recursive decomposition of graphs into clusters of sufficiently large conductance, allowing 
fast (possibly indirect) exchange of information between nodes inside clusters. The
decomposition guarantees that the number of edges between pairs of
nodes that did not exchange information decreases by a constant
fraction.
To convert the multiplicative polylogarithmic overhead for each
simulated round into the additive overhead in our final simulation
result we show connections between sparse graph spanners and
algorithms in the \GOSSIP model. This allows us to simulate known
constructions of nearly-additive sparse spanners~\cite{Pettie2009}, which then in turn can be used in our simulations for even more efficient communication.

\subsection{Our Techniques}
The key step in our approach is to devise a distributed subroutine in the
\GOSSIP model to efficiently simulate one round of the \LOCAL model by a small number of \GOSSIP rounds.
In particular, the goal is to deliver each node's current messages to all of
its neighbors, which we refer to as the \NeighborExchange problem.  
Indeed, we exhibit such an algorithm, called {\rfont Superstep}, which
requires at most $O(\log^3 n)$ rounds in the \GOSSIP model for {\it all} graphs:

\BT\label{thm:super-sketch}
\begin{sloppypar}
The {\rfont Superstep} algorithm solves \NeighborExchange in the \GOSSIP
model in $O(\log^3 n)$ rounds.
\end{sloppypar}
\ET

Our design for the {\rfont Superstep} algorithm was inspired by ideas from~\cite{CensorHillelShachnai2011}
and started with an attempt to analyze the following very natural algorithm for the \NeighborExchange problem:
In each round each node contacts a random neighbor whose message is not yet known to it. While this algorithm 
works well on most graphs, there exist graphs on
which it requires a long time to complete
due to asymmetric propagation of messages. We give an explicit example and discuss this issue in Section~\ref{sec:discussion}.

The {\rfont Superstep} algorithm is simple and operates by repeatedly performing $\log^3{n}$
rounds of the {\rfont UniformGossip} algorithm, where each node
chooses a random neighbor to contact at each round, followed by a reversal of
the message exchanges to maintain symmetry. From~\cite{ChierichettiLP2010} or its
strengthening~\cite{GiakkoupisW2011}, it is known that all pairs of vertices
(and in particular all pairs of neighbors) that lie inside a high-conductance subset
of the underlying graph exchange each other's messages within a single
iteration. An existential graph
decomposition result, given in Corollary~\ref{cor:cluster}, shows that for any
graph there is a decomposition into high-conductance clusters with at least
a constant fraction of intra-cluster edges. This implies that the number of
remaining message exchanges required decreases by a constant factor in each iteration, which
results in a logarithmic number of iterations until \NeighborExchange
is solved.

This gives a simple algorithm for solving the \Rumor problem, which
requires all nodes to receive the messages of all other nodes: By iterating {\rfont Superstep} $D$
times, where $D$ is the diameter of the network, one obtains an $O(D \cdot \log^3 n)$ round algorithm.
This is at most an $O(\log^3 n)$-factor slower than the trivial diameter lower bound and is a drastic
improvement compared to prior upper bounds~\cite{CensorHillelShachnai2011,MoskAoyamaS2006a,ChierichettiLP2010,Giakkoupis2011},
which can be of order $O(n)$ even for networks with constant or logarithmic $D$.

Beyond the \Rumor problem, it is immediate that the \NeighborExchange problem bridges the gap between the
\LOCAL and \GOSSIP models in general. Indeed, we can simply translate a single round of a \LOCAL algorithm
into the \GOSSIP model by first using any algorithm for
\NeighborExchange to achieve the local broadcast
and then performing the same local computations. We call this a simulation and
more generally define an {\it $(\alpha(G),\beta(G))$-simulator} as a transformation
that takes any algorithm in the \LOCAL model that runs in $T(G)$ rounds if the underlying topology is
$G$, and outputs an equivalent algorithm in the \GOSSIP model that runs in
$O_n(\alpha(G))\cdot T(G) + O_n(\beta)$ rounds. Thus, the simulation based on the {\rfont Superstep}
algorithm gives a $(\log^3 n,0)$-simulator.

In many natural graph classes, like graphs with bounded genus or excluded minors, one can do better.
Indeed we give a simple argument that on any (sparse) graph with {\it hereditary density} $\delta$
there is a schedule of direct message exchanges such that \NeighborExchange is achieved in $2 \delta$
rounds. Furthermore an order-optimal schedule can be computed in $\delta \log n$ rounds
of the \GOSSIP model even if $\delta$ is not known. This leads to a $(\delta,\delta \log n)$-simulator.

Another way to look at this is that communicating over any hereditary sparse graph remains fast
in the \GOSSIP model. Thus, for a general graph, if one knows a sparse subgraph that has short paths from any node to its
neighbors, one can solve the \NeighborExchange problem by communicating via these paths. Such graphs
have been intensely studied and are known as spanners. We show interesting connections between
simulators and spanners. For one, any fast algorithm for the \NeighborExchange problem induces a
sparse low-stretch spanner. The {\rfont Superstep} algorithm can thus be seen as a new
spanner construction in the \GOSSIP model with the interesting property that the total
number of messages used is at most $O(n \log^3 n)$. To our knowledge this is the first such construction.
This also implies that, in general, \NeighborExchange requires a
logarithmic number of rounds (up to $\log\log n$ factors perhaps) in the \GOSSIP model. Considering in the other direction, we show
that any fast spanner construction in the \LOCAL model can be used to further decrease
the multiplicative overhead of our $(\log^3 n,0)$-simulator. Applying this insight
to several known spanner
constructions~\cite{DGPV2008,pettie2010distributed,DMPRS2005,Pettie2009}
leads to our second main theorem:

\BT\label{thm:allsimulators}
Every algorithm in the \LOCAL model which completes in $T = T(G)$ rounds
when run on the topology $G$ can be simulated in the \GOSSIP model in
$$
O(1)\cdot\min\Big\{
  T \ \cdot \ \log^3n,\\
	T \ \cdot \ 2^{\log^* n} \log n  + \ \log^4n,\\
	T \ \cdot \ \log n                + \ 2^{\log^* n} \log^4 n,\\
	T                                 + \ \log^{O(1)}n,\\
	T \ \cdot \ \delta + \delta \log n,\\
	T \ \cdot \ \Delta  \Big\}
$$
rounds, where $n$ is the number of nodes, $\Delta$ the maximum degree and $\delta$ the hereditary density
of $G$.
\ET

When we apply this result to the greedy algorithm for the \Rumor problem,
where $T=D$, we obtain an algorithm whose $O(D + \polylog n)$ rounds are optimal up to the additive
polylogarithmic term, essentially closing the gap to the known trivial lower bound of $\Omega(D)$.

\subsection{Related Work}

The problem of spreading information in a distributed system was introduced by Demers et al.~\cite{Demers87} for the purpose of replicated database maintenance, and it has been extensively studied thereafter.   

One fundamental property of the distributed system that affects the number of rounds required for information spreading is the communication model.
The \emph{random phone call} model was introduced by Karp et al.~\cite{KarpSSV2000}, allowing every node to contact one other node in each round. In our setting, this corresponds to the complete graph. This model alone received much attention, such as in bounding the number of calls~\cite{DoerrF2010}, bounding the number of random bits used~\cite{GiakkoupisW2011}, bounding the number of bits~\cite{FraigniaudG2010}, and more.

\begin{sloppypar} The number of rounds it takes to spread information for the randomized
algorithm {\rfont UniformGossip}, in which every node chooses its
communication partner for the next round uniformly at random from its
set of neighbors, was analyzed using the conductance of the underlying
graph by Mosk-Aoyama and Shah ~\cite{MoskAoyamaS2006a}, by
Chierichetti et al.~\cite{ChierichettiLP2010}, and later by
Giakkoupis~\cite{Giakkoupis2011}, whose work currently has the best
bound in terms of conductance, of $O(\frac{\log{n}}{\Phi(G)})$ rounds,
with high probability.
\end{sloppypar}

Apart from the uniform randomized algorithm, additional algorithms were suggested for spreading information. We shortly overview some of these approaches. Doerr et al.~\cite{DoerrFS08} introduce \emph{quasi-random} rumor spreading, in which a node chooses its next communication partner by deterministically going over its list of neighbors, but the starting point of the list is chosen at random. Results are $O(\log{n})$ rounds for a complete graph and the hypercube, as well as improved complexities for other families of graphs compared to the randomized rumor spreading algorithm with uniform distribution over neighbors. This was followed by further analysis of the quasi-random algorithm~\cite{DoerrFS09,FountoulakisH2009}. A hybrid algorithm, alternating between deterministic and randomized choices~\cite{CensorHillelShachnai2011}, was shown to achieve information spreading in $O(c(\frac{\log{n}}{\Phi_c(G)}+c))$ round, w.h.p., where $\Phi_c(G)$ is the \emph{weak conductance} of the graph, a measure of connectivity of subsets in the graph.
Distance-based bounds were given for nodes placed with uniform density in $R^d$~\cite{KempeKD2001,KempeK2002}, which also address gossip-based solutions to specific problems such as resource location and minimum spanning tree.

The \LOCAL model of communication, where each node communicates with each of its neighbors in every round, was formalized by Peleg~\cite{Peleg2000}. Information spreading in this model requires a number of rounds which is equal to the diameter of the communication graph. Many other distributed tasks have been studied in this model, and below we mention a few in order to give a sense of the variety of problems studied. These include computing maximal independent sets~\cite{BarenboimE2008}, graph colorings~\cite{BarenboimE2009}, computing capacitated dominating sets~\cite{KuhnM2010}, general covering and packing problems~\cite{KuhnMW2006}, and general techniques for distributed symmetry breaking~\cite{SchneiderW2010}.


\section{Preliminaries and Definitions}

\subsection{The {\rfont UniformGossip} Algorithm}

The {\rfont UniformGossip} algorithm is a common
algorithm for \Rumor.  (It is also known as the PUSH-PULL algorithm in
some papers, such as~\cite{Giakkoupis2011}.) Initially each
vertex $u$ has some message $M_u$. At each step, every
vertex chooses a random incident edge $(u,v)$ at which point
$u$ and $v$ exchange all messages currently known to them.
The process stops when all vertices know everyone's initial messages.
In order to treat this process formally, for any fixed vertex $v$ and its
message $M_v$,
we treat the set of vertices that know $M_v$ as a set that evolves
probabilistically over time, as we explain next.

We begin by fixing an ambient graph $G=(E,V)$, which is unweighted
and directed. The {\rfont UniformGossip} process is a Markov chain
over $2^V$, the set of vertex subsets of $G$. Given a current
state $S\subseteq V$, one transition is defined as follows.
Every vertex $u$ picks an incident outgoing edge $a_u=(u,w)\in E$
uniformly at random from all such candidates. Let us call the set of
all chosen edges $A=\{a_u:u\in V\}$ an \textit{activated set}.
Further let $A^\circ=\{(u,w):(u,w)\in A\text{ or }(w,u)\in A\}$ be
the symmetric closure of $A$.
The new state of the chain is given by $S\cup B$, where by definition
a vertex $v$ is in the \textit{boundary} set $B$ if and only if
there exists $u\in S$ such that $(u,v)\in A^\circ$.
Note that $V$ is the unique absorbing state, assuming a non-empty start.

We say that an edge $(u,w)$ is \textit{activated}
if $(u,w)\in A^\circ$.
If we let $S$ model the set of nodes in possession of the message $M_v$
of some fixed vertex $v$ and we assume bidirectional message exchange along activated
edges, the new state $S\cup B$ (of the Markov process) actually describes the
set of nodes in possession of the message $M_v$ after one distributed step of
the {\rfont UniformGossip} algorithm.

Consider a $\tau$-step Markov process $K$, whose activated sets at each step
are respectively $A_1,\dots,A_\tau$. Let the \textit{reverse} of $K$,
written $K^{\rm rev}$, be the $\tau$-step
process defined by the activated sets $A_\tau,\dots,A_1$, in this order.
For a process $K$, let $K(S)$ denote the end state when started from $S$.

Without loss of generality, for our analysis we will assume that
only a single ``starting'' vertex $s$ has an initial message $M_s$. We will
be interested in analyzing the number of rounds of {\rfont UniformGossip}
that ensure that all other vertices learn $M_s$, which we call the
{\it broadcast time}. Clearly, when more than one vertex has an initial
message, the broadcast time is the same since all messages are exchanged in
parallel.

\BL[Reversal Lemma]
If $u\in K(\{w\})$, then $w\in K^{\rm rev}(\{u\})$.
\EL
In communication terms, the lemma says that if $u$ receives a message originating
at $w$ after $\tau$ rounds determined by $K$, then $w$ will receive a message
originating at $u$ after $\tau$ rounds determined by $K^{\rm rev}$.
\BPF
The condition $u\in K(\{w\})$ holds if and only if there exists a sequence
of edges $(e_{i_1},\dots,e_{i_r})$ such that
$e_{i_j}\in A_{i_j}^\circ$ for all $j$, the indices are increasing
in that $i_1<\cdots < i_r$, and
the sequence forms a path from $w$ to $u$. The presence of the
reversed sequence in $K^{\rm rev}$ implies $w\in K^{\rm rev}(\{u\})$.
\EPF


\subsection{Conductance}
The notion of \emph{graph conductance} was introduced by Sinclair~\cite{Sinclair93}. We require a
more general version, which we introduce here.  We begin with the requisite notation on
edge-weighted graphs. We assume that each edge $(u,v)$ has a weight $w_{uv} \in [0,1]$.  For an
unweighted graph $G =(V,E)$ and any $u,v \in V$, we define $w_{uv} = 1$ if $(u,v) \in E$ and $w_{uv} =
0$ otherwise.  Now we set $w(S,T)=\sum_{u\in S,v\in T} w_{uv}$. Note that in this definition it need
not be the case that $S\cap T = \emptyset$, so, e.g., $w(S,S)$, when applied to an unweighted graph,
counts every edge in $S$ twice.
The volume of a set $S\subseteq V$ with respect to $V$ is written as $\vol(S)=w(S,V)$. Sometimes we
will have different graphs defined over the same vertex set. In such cases, we will write the
identity of the graph as a subscript, as in $\vol_G(S)$, in order to clarify which is the ambient
graph (and hence the ambient edge set).  Further, we allow self-loops at the vertices.  A single
loop at $v$ of weight $\alpha$ is modeled by setting $w_{vv}=2\alpha$, because both ends of the edge
contribute $\alpha$.

For a graph $G=(V,E)$ and a cut $(S,T)$ where $S,T\subseteq V$ and $S\cap T=\emptyset $
(but where $T\cup S$ does not necessarily equal all of $V$),
the \textit{cut conductance} is given by
\beq \label{eq:cutcond}
\varphi(S,T)
	= \frac{w(S,T)}{\min\big\{\vol_G(S),\vol_G(T)\big\}}.
\eeq


For a subset $H\subseteq V$ we need to define the \textit{conductance of $H$ (embedded) in $V$}.
We will use this quantity to measure how quickly the {\rfont
  UniformGossip} algorithm
proceeds in $H$, while accounting for the
fact that edges in $(H,V- H)$ may slow down the process.
The conductance of $H$ in $G$ is defined by
\beq
\Phi(H)=\min_{S\subseteq H} \varphi(S,H- S)
\eeq

Note that the classical notion of conductance of $G$ (according to Sinclair~\cite{Sinclair93})
equals $\Phi(V)$ in our notation.

When we want to explicitly emphasize the ambient graph $G$ within which $H$
resides, we will write $\Phi_G(H)$.

A few arguments in this paper will benefit from the notion of a ``strongly induced'' graph of a vertex subset of
an ambient graph $G$.

\BD
Let $U\subseteq V$ be a vertex subset of $G$. The \textit{strongly induced} graph
of $U$ in $G$ is a (new) graph $H$ with vertex set $U$, whose
edge weight function $h:U\times U\to \R$ is defined by
\begin{align*}
h_{uv} =
\begin{cases}
w_{uv},\quad &\text{if $u\neq v$}, \\
w_{uu}+\sum_{x\in V- U}w_{ux},\quad &\text{if $u=v$.}
\end{cases}
\end{align*}
\ED

Note that by construction we have $\Phi_H(U) = \Phi_G(U)$.
The significance of this notion is the fact that the Markov process, describing
the vertex set in possession of some message $M_s$ for a starting vertex
$s\in U$ in the {\rfont UniformGossip} algorithm executed
on the strongly induced $H$, behaves identically to the respective process
in $G$ observed only on $U$. In particular, this definition allows us to use Theorem 1
of~\cite{Giakkoupis2011} in the following form:
\begin{theorem}\label{thm:grumor}
For any graph $G=(V,E)$ and a subgraph $U\subseteq V$ and any start vertex in $U$,
the broadcast time of the {\rfont UniformGossip} algorithm on $U$ is $O(\Phi_G(U)^{-1}\log|U|)$ rounds w.h.p.
\end{theorem}

\section{Solving \NeighborExchange in $O(\log^3 n)$ Rounds} \label{sec:super}
The idea behind our algorithm for solving the \NeighborExchange problem is as follows. For every graph there exists a partition into clusters
whose conductance is high, and therefore the {\rfont UniformGossip}
algorithm allows information to spread quickly in each cluster. The latter further
implies that pairs of neighbors inside a cluster exchange their messages
quickly  (perhaps indirectly). What remains is to exchange messages across inter-cluster edges. This is done recursively.
In the following subsection we describe the conductance decomposition and then in Subsection~\ref{sec:alg} we give the details for the algorithm together with the proof of correctness.

\subsection{Conductance Decomposition of a Graph}
As described, our first goal is to partition the graph into clusters
with large conductance. The challenge here is to do so while limiting
the number of inter-cluster edges, so that we can efficiently apply this argument 
recursively. (Otherwise, this could be trivially done
in any graph, for example by having each node as a separate cluster.)
We are going to achieve this in the following lemma whose proof (found in Appendix~\ref{app:balcut}) is
very similar to that of Theorem 7.1 in~\cite{SpectralSparse}. Note that
for our eventual algorithm, we are only going to need an existential
proof of this clustering and not an actual algorithm for finding it.

\BL \label{lemma:balcut}
Let $S\subseteq V$ be of maximum volume such that $\vol(S)\le \vol(V)/2$
and $\varphi(S,V- S)\le \xi$, for a fixed parameter $\xi\ge\Phi(G)$.
If $\vol(S)\le\vol(V)/4$, then $\Phi(V- S) \ge \xi/3$.
\EL

Lemma~\ref{lemma:balcut} says that if a graph has no sparse balanced cuts,
then it has a large subgraph which has no sparse cuts.
The following corollary establishes that Lemma~\ref{lemma:balcut} holds even
in the case when the ambient graph is itself a subgraph of a larger graph.

\BC \label{cor:balcut}
Let $U \subseteq V$ and let $S \subseteq U$ be of maximum volume such that $\vol(S)\le \vol(U)/2$
and $\varphi(S,U- S)\le \xi$, for a fixed parameter $\xi\ge\Phi(U)$.
If $\vol(S)\le\vol(U)/4$, then $\Phi(U- S) \ge \xi/3$.
\EC

\BPF
Observe that the proof of Lemma~\ref{lemma:balcut} holds when
the graph has loops, i.e. $w_{uu}\neq 0$ for some $u$'s.
Let $H$ be the strongly induced graph of $U$.
It follows from the definition that for any two disjoint sets
$A,B\subseteq U$ we have $\vol_G(A)=\vol_H(A)$ and $w(A,B)=h(A,B)$. We can therefore apply
Lemma~\ref{lemma:balcut} to $H$ and deduce that the statement holds
for the respective sets in $G$.
\EPF

We are now ready to state and analyze the strong clustering algorithm. We emphasize that this is not a distributed algorithm, but an algorithm that only serves as a proof of existence of the partition.
First, consider the following subroutine:
\begin{quote}
\small{
\sl
\texttt{Cluster($G$,$U$,$\xi$):}
\begin{itemize}
\item[]
The inputs are a graph $G=(V,E)$, a subset $U\subseteq V$ and a parameter $0<\xi<1$.
\item[1.]
Find a subset $S\subseteq U$ of maximum volume such that
$\vol(S)\le\vol(U)/2$ and $\varphi(S,U- S)\le \xi$.
\item[2.]
If no such $S$ exists, then stop and output a single cluster $\{U\}$. Otherwise,
\item[3a.]
If $\vol(S)\le\vol(U)/4$, output $\{U- S\}\cup\text{\tt Cluster($G$,$S$,$\xi$)}$.
\item[3b.]
If $\vol(S)>\vol(U)/4$, output
$\text{\tt Cluster($G$,$S$,$\xi$)}\cup\text{\tt Cluster($G$,$U- S$,$\xi$)}$.
\end{itemize}
}
\end{quote}
The clustering algorithm for a graph $G=(V,E)$ is simply a call to \texttt{Cluster($G$,$V$,$\xi$)}.
The following theorem is proven in Appendix~\ref{app:clstr}.

\BT\label{thm:clstr}
For every $0<\zeta<1$, every graph $G=(V,E)$ with edge weights
$w_{uv}\in\{0\}\cup[1,+\infty)$ has a partition
$V=V_1\cup\cdots\cup V_k$ such that
$\Phi(V_i)\ge \frac{\zeta}{\log_{4/3}\vol(V)}$, for all $i$, and
$\sum_{i<j}w(V_i,V_j)\le \frac{3\zeta}{2}\vol(V)$.
\ET

In this paper, we are going to use the following specialization of this theorem, obtained by plugging in $\zeta = 1/3$:
\BC\label{cor:cluster}
Every unweighted graph on $m$ edges has a clustering
that cuts at most $\frac{m}{2}$ edges and each cluster has conductance at least
$\frac{1}{3\log_{4/3} 2m}$.
\EC

\subsection{The {\tt Superstep} Algorithm for the \NeighborExchange Problem}  \label{sec:alg}

In this section, we will the describe the {\tt Superstep} algorithm, which
solves the \NeighborExchange problem. Recall that, for this problem, all vertices
$v$ are assumed to possess an initial message $M_v$, and the goal is for every pair of neighbors
to know each other's initial messages.

We now describe our communication protocol, which specifies a local,
per-vertex rule that tells a node which
edge to choose for communication at any given round. It is assumed
that the node will greedily transmit all
messages known to it whenever an edge is chosen for communication.
The protocol described here
will employ some auxiliary messages, which are needed exclusively
for its internal workings.

The {\tt Superstep} subroutine described in this section is
designed to ensure that, after a single invocation, all neighbors $(u,w)$
in an undirected graph $G$ have exchanged each other's initial messages.
Clearly then, $D$ invocations of {\tt Superstep}, where $D$
is the diameter of $G$, ensure that a message starting at
vertex $v$ reaches all $u\in V$, and this holds for all messages.
$D$ invocations of {\tt Superstep} thus resolve the \Rumor problem.

If $E$ is a set of undirected edges, 
let $\vec{E}=\{(u,w):\{u,w\}\in E\}$ be the corresponding directed graph.

\begin{quote}
\small{
\sl
\texttt{Superstep($G$,$\tau$):}\\
The parameter $G=(V,E)$ is an unweighted, undirected graph,
and $\tau$ is a positive integer.\\
Set $F_0 := \vec{E}$ and $i:=0$. While $F_i \neq \emptyset$, repeat:
\begin{itemize}
\item[1.] {\it (First half)}
\begin{itemize}
	\item[1a.] Initialize every vertex $v$ with a new auxiliary
	message $a(v)$, unique to $v$. (This messages is added to the set of
	initial messages that $v$ happens to know currently.)
	\item[1b.] Perform the {\rfont UniformGossip} algorithm with respect to
	$F_i$ for $\tau$ rounds. And denote the outcome of the random activated edge
	choices by $K_i$
	\item[1c.] For every vertex $u$ and neighbor $w$,
	let $X_{uw}$ be the indicator that $u$ received $a(w)$
\end{itemize}
\item[2.] {\it (Second half)}
\begin{itemize}
	\item[2a.] Initialize every vertex $v$ with a fresh auxiliary
	message $b(v)$, unique to $v$
	\item[2b.] Perform $K_i^{\rm rev}$, the reverse process of the
	one realized in Step 1b
	\item[2c.] For every vertex $u$ and neighbor $w$,
	let $Y_{uw}$ be the indicator that $u$ received $b(w)$
\end{itemize}
	\item[3.] \textit{(Pruning)}
	Compute the set of pruned \textit{directed} edges
$P_i=\big\{(u,w):X_{uw}+Y_{uw} > 0\big\}$
\item[4.]
Set $F_{i+1}:=F_i-P_i$ and $i := i+1$
\end{itemize}
}
\end{quote}
It is easily verified that the above algorithm can be implemented
in the \GOSSIP model of communication.

\BT\label{thm:super}
Let $G=(V,E)$ be an undirected, unweighted graph with $|V|=n$ and $|E|=m$.
Then, after one invocation of {\rm{\tt Superstep($G$,$\tau$)}},
where $\tau=\Theta\big(\log^2 m\big)$,
the following hold
with probability $1-1/n^{\Omega(1)}$:
\begin{itemize}
\item[(i)] Every pair of neighbors $\{u, w\}\in E$ receive each other's messages.
\item[(ii)] The algorithm performs
	$\Theta\big(\log^3 m\big)$ distributed rounds.
\end{itemize}
\ET

Our proof of Theorem~\ref{thm:super} is structured as follows.
Let $\vec{E}=F_0,\dots, F_d=\emptyset$ be the respective edge sets of each
iteration in {\tt Superstep}. We are going to show that, with probability $1-1/n^{\Omega(1)}$,
the following invariants
are maintained at each iteration:
\begin{itemize}
\item[(a)] The directed edge set $F_i$ is symmetric in the sense that
$(u,w)\in F_i \Rightarrow (w,u)\in F_i$,
\item[(b)] The size of $F_i$ reduces by a constant factor at each iteration.
Formally,
$\vol(F_{i+1})\le \frac{1}{2} \vol(F_i)$, and
\item[(c)] After the $i$-th iteration, for every $(u,w)\in \vec{E}-F_{i+1}$, vertex $u$ has
received the message of vertex $w$ and vice-versa.
\end{itemize}

Since $F_d=\emptyset$, claim (c) implies part (i) of Theorem~\ref{thm:super}.
Claim (b) implies that the maximum number of iterations is $\log 2m$. Noting that
every iteration entails $2\tau$ distributed rounds, establishes part (ii) of
Theorem~\ref{thm:super}.

\BPF[Proof of Claim ({\rm a}):]
Initially, $F_0$ is symmetric by construction. Inductively, assume that $F_i$ is
symmetric. The Reversal Lemma applied to $K_i$ and $K_i^{\rm rev}$ implies $X_{uw}=Y_{wu}$, for
all $u,w\in V$. This in turn implies that
$X_{uw}+Y_{uw}=X_{wu}+Y_{wu}$, so $P_i$ is symmetric.
Since $F_i$ is symmetric by hypothesis,
we can conclude that $F_{i+1}=F_i-P_i$ is symmetric as well.
\EPF

\BPF[Proof of Claim ({\rm b}):]
Consider the graph $G_i=(V,F_i)$ on the edge set $F_i$. Since $F_i$ is symmetric, by Claim (a),
we can treat $G_i$ as undirected for the purposes of analyzing the { \rfont UniformGossip} algorithm.
Let $V_1\cup\cdots\cup V_k$ be the decomposition of $G_i$ promised by Corollary~\ref{cor:cluster}.
(Note that the corollary holds for disconnected graphs, which may arise.)
We thus have $\Phi(V_j)\ge \frac{1}{3\log {4/3} 2m}$, for all $1\le j\le k$.

The choice $\tau=O\big(3\log_{4/3}2m\cdot\log m\big)$ ensures,
via Theorem~\ref{thm:grumor}, that the first { \rfont UniformGossip} execution in every iteration
mixes on all $V_j$ with probability $1-1/n^{\Omega(1)}$. Mixing in $V_j$ implies that
for every internal edge $(u,w)$, where $u,w\in V_j$ and $(u,w)\in F_i$, the
vertices $(u,w)$ receive each other's auxiliary messages. The latter is summarized
as $X_{uw}=X_{wu}=1$. Applying the Reversal Lemma to the second execution of the {\rfont UniformGossip} algorithm,
we deduce that $Y_{uw}=Y_{wu}=1$ as well.
These two equalities imply, by the definition of $P_i$, that
$P_i$ is a superset of the edges not cut by the decomposition $V_1\cup\cdots\cup V_k$.
Equivalently, $F_{i+1}$ is a subset of the cut edges. Corollary~\ref{cor:cluster}, however,
bounds the volume of the cut edges by $\frac{1}{2}\vol(F_i)$, which concludes the proof of
Claim (b).
\EPF

\BPF[Proof of Claim ({\rm c}):]
Initially, $\vec{E}-F_0=\emptyset$ and so the claim holds trivially. By induction,
the claim holds for edges in $\vec{E}-F_i$. And so it suffices to establish that
$u$ and $v$ exchange their respective payload messages for all $(u,w)\in P_i$. However, this
is equivalent to the conditions $X_{uw}+Y_{uw}>0$, which are enforced by the
definition of $P_i$.
\EPF

Finally, our main result, Theorem~\ref{thm:super-sketch}, follows as a corollary of Theorem~\ref{thm:super}.

\section{Solving \NeighborExchange in Hereditary Sparse Graphs} \label{sec:sparse}

Now we ask what can be achieved if instead of
exchanging information indirectly as done in the {\rfont Superstep} algorithm, we exchange information only directly between neighbors. We will show in this section that this results in very simple deterministic algorithms for an important class of graphs that includes bounded genus graphs and all graphs that can be characterized by excluded minors~\cite{mader67minor,minor}.  The results here will be used for the more general simulators in Section~\ref{sec:simulators}.

As before we will focus on solving the \NeighborExchange problem. One trivial way to solve this problem is for each node to contact its neighbors directly, e.g., by using a simple round robin method. This takes at most $\Delta$ time, where $\Delta$ is the maximum-degree of the network. However, in some cases direct message exchanges work better. One graph that exemplifies this is the star graph on $n$ nodes. While it takes $\Delta = n$ time to complete a round robin in the center, after just a single round of message exchanges each leaf has initiated a bidirectional link to the center and thus exchanged its messages. On the other hand, scheduling edges cannot be fast on dense graphs with many more edges than nodes. The following lemma shows that the \emph{hereditary density} captures how efficient direct message exchanges can be on a given graph:

\begin{lemma}\label{lem:sched}
Let the hereditary density $\delta$ of a graph $G$ be the minimal integer such that for every subset of nodes $S$ the subgraph induced by $S$ has at most density $\delta$, i.e., at most $\delta|S|$ edges.
\begin{enumerate}
	\item Any schedule of direct message exchanges that solves the \NeighborExchange problem on $G$ takes at least $\delta$ rounds.
	\item There exists a schedule of the edges of $G$ such each node needs only $2 \delta$ direct message exchanges to solve the \NeighborExchange problem.
\end{enumerate}
\end{lemma}
\begin{proof}
Since the hereditary density of $G$ is $\delta$, there is a subset of nodes $S \subseteq V$ with at least $\delta|S|$ edges between nodes in $S$. In each round, each of the $|S|$ nodes is allowed to schedule at most one message exchange, so a simple pigeonhole principle argument shows that at least one node needs to initiate at least $\delta$ message exchanges.

For the second claim, we are going to show that for any $\eps >0$ there is an $O(\eps^{-1}\log{n})$-time deterministic distributed algorithm in the \LOCAL model that assigns the edges of $G$ to nodes such that each node is assigned at most $2(1+\eps)\delta$ edges. Then setting $\eps < (3\delta)^{-1}$ makes the algorithm inefficient but finishes the existential proof. 

The algorithm runs in phases in which, iteratively, a node takes responsibility for some of the remaining edges connected to it. All edges that are assigned are then eliminated and so are nodes that have no unassigned incident edges. In each phase, every node of degree at most $2(1+\eps)\delta$ takes responsibility for all of its incident edges (breaking ties arbitrarily). At least a $1/(1+\frac{1}{\eps})$ fraction of the remaining nodes fall under this category in every phase. This is because otherwise, the number of edges in the subgraph would be more than $(|S|-|S|/(1+\frac{1}{\eps}))(2(1+\eps)\delta)/2=|S|\delta$, which would contradict the fact thatthe hereditary densityof the graph equals $\delta$ of. What remains after each phase is an induced subgraph which, by definition of the hereditary density, continues to have hereditary density at most $\delta$. The number of remaining nodes thus decreases by a factor of $1-1/(1+\frac{1}{\eps})$ in every phase and it takes at most $O(\log_{1+\eps}{n})$ phases until no more nodes remain, at which point all edges have been assigned to a node.
\end{proof}

We note that the lower bound of Lemma~\ref{lem:sched} is tight in all
graphs, i.e., the upper bound of $2\delta$ can be improved to
$\delta$. Graphs with hereditary density $\delta$, also known as
$(0,\delta)$-sparse graphs, are thus exactly the graphs in which
$\delta$ is the minimum number such that the edges can be oriented to
form a directed graph with outdegree at most $\delta$. This in turn is
equivalent to the \emph{pseudoarboricity} of the graph, i.e., the
minimum number of pseudoforests needed to cover the graph. Due to the
matroid structure of pseudoforests, the pseudoarboricity can be
computed in polynomial time. For our purposes the (non-distributed)
algorithms to compute these optimal direct message exchange schedule
are too slow. Instead, we present a simple and fast algorithm, based on
the \LOCAL algorithm in Lemma~\ref{lem:sched}, which computes a
schedule that is within a factor of $2 + \eps$ of the optimal. We note that the {\rfont DirectExchange} algorithm presented here works in the \GOSSIP model and furthermore does not require the hereditary density $\delta$ to be known \emph{a priori}.
The following is the algorithm for an individual node $v$:

\begin{quote}
\small{
\sl
\noindent \texttt{DirectExchange:}\\
Set $\delta' = 1$ and $H=\emptyset$. $H$ is the subset of neighbors in $\Gamma(v)$ that node $v$ has
exchanged messages with.
Repeat:
\begin{itemize}
\item[]
    $\delta' = (1+ \eps) \delta'$
\item[]
    for $O(\frac{1}{\eps}\cdot\log n)$ rounds do
    \begin{itemize}
    \item[]
        if $|\Gamma(v)\setminus H| \leq \delta'$
        \begin{itemize}
        \item[]
        		during the next $\delta'$ rounds exchange
                        messages with all neighbors in
                        $\Gamma(v)\setminus H$
        \item[]
        		terminate
        \end{itemize}
        \item[]
        else
        \begin{itemize}
        \item[]
            wait for $\delta'$ rounds
        \end{itemize}
        \item[]
        update $H$
    \end{itemize}
\end{itemize}
}
\end{quote}

\begin{theorem}
For any constant $\eps > 0$, the deterministic algorithm {\rfont DirectExchange}
solves the \NeighborExchange problem in the \GOSSIP model using $O(\frac{\delta\log{n}}{\eps^2})$ rounds, where
$\delta$ is the hereditary density of the underlying topology. During
the algorithm, each node initiates at most $2(1+\eps)^2 \delta$
exchanges.
\end{theorem}
\begin{proof}
Let $\delta$ be the hereditary density of the underlying topology. We
know from the proof of Lemma~\ref{lem:sched} that the algorithm
terminates during the for-loop if $\delta'$ is at least
$2(1+\eps)\delta$. Thus, when the algorithm terminates, $\delta'$ is at
most $2(1+\eps)^2\delta$ which is also an upper bound on the number
of neighbors contacted by any node. In the $(i+1)^{\mathrm{th}}$-to-last iteration of the outer loop,
$\delta'$ is at most $2(1+\eps)^2\delta / (1 + \eps)^i$, and the
running time for this phase is thus at most $2(1+\eps)^2\delta /(1 + \eps)^i \cdot
O(\frac{1}{\eps} \log n)$. Summing up over these
powers of $1/(1+\eps)$
results in a total of at most $\delta/((1+\eps)-1) \cdot O(\frac{1}{\eps} \log
n) = O(\frac{\delta\log{n}}{\eps^2})$ rounds.
\end{proof}

\section{Simulators and Graph Spanners}\label{sec:simulators}


In this section we generalize our results to arbitrary simulations of \LOCAL algorithms in the \GOSSIP model and point out connections to graph spanners, another well-studied subject.

Recall that we defined the \NeighborExchange problem exactly in such a way that it simulates in the \GOSSIP model what is done in one round of the \LOCAL model. With our solutions, an $O(\delta \log n)$-round algorithm and an $O(\log^3 n)$-round algorithm for the \NeighborExchange problem in the \GOSSIP model, it is obvious that we can now easily convert any $T$-round algorithm for the \LOCAL model to an algorithm in the \GOSSIP model, e.g., by $T$ times applying the {\rfont Superstep} algorithm. In the case of the {\rfont DirectExchange} algorithm we can do even better. While it takes $O(\delta \log n)$ rounds to compute a good scheduling, once it is known it can be reused and each node can simply exchange messages with the same $O(\delta)$ nodes without incurring an additional overhead. Thus, simulating the second and any further rounds can be easily done in $O(\delta)$ rounds in the \GOSSIP model. This means that any algorithm that takes $O(T)$ rounds to complete in the \LOCAL model can be converted to an algorithm that takes $O(\delta T + \delta \log n)$ rounds in the \GOSSIP model. We call this a simulation and define simulators formally as follows.

\begin{definition}
An $(\alpha, \beta)$-simulator is a way to transform any algorithm $A$ in the \LOCAL model to an algorithm $A'$ in the \GOSSIP model such that $A'$ computes the same output as $A$ and if $A$ takes $O(T)$ rounds than $A'$ takes at most $O(\alpha T + \beta)$ rounds.
\end{definition}

Phrasing our results from Section~\ref{sec:alg} and Section~\ref{sec:sparse} in terms of simulators we get the following corollary.

\begin{corollary}\label{cor:simplesimulators}
For a graph $G$ of $n$ nodes, hereditary density $\delta$, and maximum degree $\Delta$, the following hold:
\begin{itemize}
	\item There is a randomized $(\log^3 n,0)$-simulator.
	\item There is a deterministic $(\Delta,0)$-simulator.
	\item There is a deterministic $(2(1+\eps)^2\delta,O(\delta \eps^{-2} \log n))$-simulator for any $\epsilon > 0$ or, simply, there is a $(\delta,\delta \log n)$-simulator.
\end{itemize}
\end{corollary}

\begin{sloppypar}Note that for computations that require many rounds in
  the \LOCAL model the $(2(1+\eps)^2\delta,O(\delta \eps^{-2} \log
  n))$-simulator is a $\log n$-factor faster than repeatedly applying
  the {\rfont DirectExchange} algorithm. This raises the question
  whether we can similarly improve our $(\log^3 n,0)$-simulator to
  obtain a smaller multiplicative overhead for the simulation.
\end{sloppypar}

What we would need for this is to compute, e.g., using the {\rfont Superstep} algorithm, a schedule that can then be repeated to exchange messages between every node and its neighbors. What we are essentially asking for is a short sequence of neighbors for each node over which each node can indirectly get in contact with all its neighbors. Note that any such schedule of length $t$ must at least fulfill the property that the union of all edges used by any node is connected (if the original graph $G$ is connected) and even more that each node is connected to all its neighbors via a path of length at most $t$. Subgraphs with this property are called \emph{spanners}. Spanners are well-studied objects, due to their extremely useful property that they approximately preserve distances while potentially being much sparser than the original graph. The quality of a spanner is described by two parameters, its number of edges and its \emph{stretch}, which measures how well it preserves distances.

\begin{definition}[Spanners]
A subgraph $S=(V,E')$ of a graph $G = (V,E)$ is called an \emph{$(\alpha, \beta)$-stretch spanner} if any two nodes $u,v$ with distance $d$ in $G$ have distance at most $\alpha d + \beta$ in $S$.
\end{definition}

From the discussion above it is also clear that any solution to the \NeighborExchange problem in the \GOSSIP model also computes a spanner as a byproduct.

\begin{lemma}\label{lem:simulatorimpliesspanner}
If $A$ is an algorithm in the \GOSSIP model that solves the \NeighborExchange problem in any graph $G$ in $T$ rounds then this algorithm can be used to compute a $(T,0)$-stretch spanner with hereditary density $T$ in $O(T)$ rounds in the \GOSSIP model.
\end{lemma}

While there are spanners with better properties than the $(\log^3 n, 0)$-stretch and $\log^3 n$-density implied by Lemma~\ref{lem:simulatorimpliesspanner} and Theorem~\ref{thm:super}, our construction has the interesting property that the number of messages exchanged during the algorithm is at most $O(n \log^3 n)$, whereas all prior algorithms rely on the broadcast nature of the \LOCAL model and therefore use already $O(n^2)$ messages in one round on a dense graph. Lemma~\ref{lem:simulatorimpliesspanner} furthermore implies a nearly logarithmic lower bound on the time that is needed in the \GOSSIP model to solve the \NeighborExchange problem:

\begin{corollary} \label{cor:lowerbound}
For any algorithm in the \GOSSIP model that solves the \NeighborExchange problem there is a graph $G$ on $n$ nodes on which this algorithm takes at least $\omega(\frac{\log n}{\log \log n})$ rounds,.
\end{corollary}
\begin{proof}
Assume an algorithm takes at most $T(n)$ rounds on any graph with $n$ nodes. The edges used by the algorithm form a $T(n)$-stretch spanner with density $T(n)$, as stated in Lemma~\ref{lem:simulatorimpliesspanner}. For values of $T(n)$ which are too small it is known that such spanners do not exist~\cite{peleg1989spanners}. More specifically it is known that there are graphs with $n$ nodes, density at least $1/4 n^{1/r}$ and girth $r$, i.e., the length of the smallest cycle is $r$. In such a graph any $(r-2)$-stretch spanner has to be the original graph itself, since removing a single edge causes its end-points to have distance at least $r-1$, and thus the spanner also have density $1/4 n^{1/r}$. Therefore $T(n) \geq \argmin_r \{r-2,1/4n^{1/r}\} = \omega(\frac{\log n}{\log \log n})$.
\end{proof}

Interestingly, it is not only the case that efficient simulators imply good spanners but the next lemma shows as a converse that good existing spanner constructions for the \LOCAL model can be used to improve the performance of simulators.

\begin{theorem}\label{thm:spannerimprovement}
If there is an algorithm that computes an $(\alpha,\beta)$-stretch spanner with hereditary density $\delta$ in $O(T)$ rounds in the \LOCAL model than this can be combined with an $(\alpha', \beta')$-simulator to an $(\alpha \delta, T \alpha' + \beta' + \delta \log n + \delta \beta)$-simulator.
\end{theorem}
\begin{proof}
For simplicity we first assume that $\beta = 0$, i.e., the spanner $S$ computed by the algorithm in the \LOCAL model has purely multiplicative stretch $\alpha$ and hereditary density $\delta$. Our strategy is simple: We are first going to compute the good spanner by simulating the spanner creation algorithm from the \LOCAL model using the given simulator. This takes $T \alpha' + \beta'$ rounds in the \GOSSIP model. Once this spanner $S$ is computed we are only going to communicate via the edges in this spanner. Note that for any node there is a path of length at most $\alpha$ to any of its neighbors. Thus if we perform $\alpha$ rounds of \LOCAL-flooding rounds in which each node forwards all messages it knows of to all its neighbors in $S$ each node obtains the messages of all its neighbors in $G$. This corresponds exactly to a \NeighborExchange in $G$. Therefore if we want to simulate $T'$ rounds of an algorithm $A$ in the \LOCAL model on $G$ we can alternatively perform $\alpha T$ \LOCAL computation rounds on $S$ while doing the \LOCAL computations of $A$ every $\alpha$ rounds. This is a computation in the \LOCAL model but on a sparse graph. We are therefore going to use the $(O(\delta),O(\delta \log n))$-simulator from Corollary~\ref{cor:simplesimulators} to simulate this computation which takes $O(\delta \alpha T' + \delta \log n)$ rounds in the \GOSSIP model. Putting this together with the $T \alpha' + \beta'$ rounds it takes to compute the spanner $S$ we end up with $\delta \alpha T' + \delta \log n + T \alpha' + \beta'$ rounds in total.

\begin{sloppypar}In general (i.e., for $\beta > \alpha$) it is not
  possible (see, e.g., Corollary~\ref{cor:lowerbound}) to simulate the
  \LOCAL algorithm step by step. Instead we rely on the fact that any
  \LOCAL computation over $T$ rounds can be performed by each node
  first gathering information of all nodes in a $T$-neighborhood and
  then doing \LOCAL computations to determine the output. For this all
  nodes simply include all their initial knowledge (and for a
  randomized algorithm all the random bits they might use throughout
  the algorithm) in a message and flood this in $T$ rounds to all node
  in their $T$-neighborhood. Because a node now knows all information
  that can influence its output over a $T$-round computation it can
  now locally simulate the algorithm for itself and its neighbors to
  the extend that its output can be determined. Having this we
  simulate the transformed algorithm as before: We first precomute $S$
  in $T \alpha' + \beta'$ time and then simulate the $T'$ rounds of
  flooding in $G$ by performing $\alpha T' + \beta$ rounds of
  \LOCAL-flooding in $S$. Using the $(O(\delta),O(\delta \log
  n))$-simulator this takes $O(\delta (\alpha T' + \beta) + \delta
  \log n)$ rounds in the \GOSSIP model.
\end{sloppypar}
\end{proof}

\begin{corollary}\label{cor:allsimulators}
\begin{sloppypar}
There is a $(2^{\log^* n} \log n, \log^4 n)$-simulator, a $(\log n,2^{\log^* n} \log^4 n)$-simulator and a $(O(1),\polylog n)$-simulator.
\end{sloppypar}
\end{corollary}
\begin{proof}
We are going to construct the simulators with increasingly better multiplicative overhead by applying Theorem~\ref{thm:spannerimprovement} to existing spanner constructions~\cite{DGPV2008,pettie2010distributed,DMPRS2005,Pettie2009} for the \LOCAL model.
We first construct a $(\log^2 n,\log^4 n)$-simulator by combining our new $(\log^3 n,0)$-simulator with the deterministic spanner construction in~\cite{DGPV2008}. The construction in~\cite{DGPV2008} takes $O(\log n)$ rounds in the \LOCAL model and adds at most one edge to each node per round. Using $\alpha = T = \delta=O(\log n)$, $\alpha'=\log^3 n$ and $\beta = \beta' = 0$ in Theorem~\ref{thm:spannerimprovement} gives the desired $(\log^2 n,\log^4 n)$-simulator.
Having this simulator, we can use~\cite{pettie2010distributed} to improve the multiplicative overhead while keeping the additive simulation overhead the same.
In~\cite{pettie2010distributed} an $\alpha = (2^{\log^* n} \log n)$-stretch spanner with constant hereditary density $\delta = O(1)$ is constructed in $T = O(2^{\log^* n} \log n)$-time in the \LOCAL model. Using these parameters and the $(\log^2 n,\log^4 n)$-simulator in Theorem~\ref{thm:spannerimprovement} leads to the strictly better $(2^{\log^* n} \log n, \log^4 n)$-simulator claimed here.
Having this simulator, we can use it with the randomized spanner construction in~\cite{DMPRS2005}. There, an $\alpha$-stretch spanner, with $\alpha = O(\log n)$, is constructed in $T = O(\log^3 n)$-time in the \LOCAL model by extracting a subgraph with $\Omega(\log n)$ girth. Such a graph has constant hereditary density $\delta = O(1)$, as argued in~\cite{peleg1989spanners}. Using these parameters and the $(2^{\log^* n} \log n,\log^4 n)$-simulator in Theorem~\ref{thm:spannerimprovement} leads to the $(\log n,2^{\log^* n} \log^4 n)$-simulator.
Finally, we can use any of these simulators together with the
nearly-additive $(5+\eps,\polylog n)$-spanner construction
from~\cite{Pettie2009} to obtain our last simulator. It is easy to
verify that the randomized construction named $AD^{\log \log n}$ in~\cite{Pettie2009} can be computed in a distributed fashion in the \LOCAL model in $\polylog n$ time and has hereditary density $\delta = O(1)$. This together with any of the previous simulators and Theorem~\ref{thm:spannerimprovement} results in a $(O(1),\polylog n)$-simulator.
\end{proof}

With these various simulators it is possible to simulate a computation
in the \LOCAL model with very little (polylogarithmic) multiplicative
or additive overhead in the \GOSSIP model. Note that while the
complexity of the presented simulators is incomparable, one can
interleave their executions (or the executions of the simulated
algorithms) and thus get the best runtime for any instance. This,
together with the results from Corollary~\ref{cor:allsimulators}
and~\ref{cor:simplesimulators}, proves our main result of Theorem~\ref{thm:allsimulators}.

\section{Discussion}
\label{sec:discussion}  
This paper presents a more efficient alternative to the {\rfont UniformGossip} algorithm that allows fast rumor spreading on all graphs, with no dependence on their conductance. We then show how this leads to fast simulation in the \GOSSIP model of any algorithm designed for the \LOCAL model by  constructing sparse spanners. This work leaves some interesting directions for future work, which we discuss below.

First, as mentioned in the introduction, there are cases in which the algorithm where each node chooses a neighbor uniformly at random only from among those it has not yet heard from (directly or indirectly), performs slower than the optimal. An example is the graph in Figure~\ref{fig:example}, where $C_i$ stands for a clique of size $O(1)$ in which every node is also connected to the node $x_i$.
\begin{figure}
\label{fig:example}
\begin{center}
\includegraphics[trim = 5cm 9cm 5cm 5cm, clip, scale=0.6]{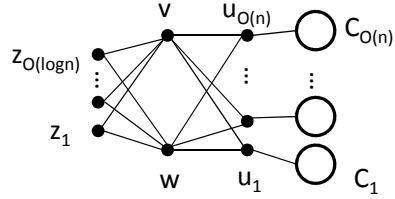}
\end{center}
\caption{An example illustrating the behavior of the algorithm choosing a random neighbor whose information is still unknown.}
\end{figure}
In this example, it takes $2$ rounds for the node $w$ to hear about the node $v$ (through nodes in $\{z_1,\dots,z_{O(\log{n})}\}$). During these rounds there is a high probability that a constant fraction of the nodes in $\{u_1,\dots,u_{O(n)}\}$ did not yet hear from neither $v$ nor $w$. With high probability, a constant fraction of these will contact $w$ before contacting $v$, after which they will not contact $v$ anymore because they will have heard from it through $w$. This leaves $O(n)$ nodes which $v$ has to contact directly (since nodes in $\{z_1,\dots,z_{O(\log{n})}\}$ are no longer active since they already heard from both of their neighbors), resulting in a linear number of rounds for \NeighborExchange.

We note, however, that this specific example can be solved by requiring nodes that have heard from all their neighbors to continue the algorithm after resetting their state, in the sense that they now consider all their neighbors to be such that they have not heard from (this is only for the sake of choosing the next neighbor to contact, the messages they send can include previous information they received). Therefore, we do not rule out the possibility that this algorithm works well, but our example suggests that this may not be trivial to prove.

Second, regarding our solution to the \Rumor problem, the {\rfont Superstep} algorithm, as presented, can be implemented in synchronous environments in a straightforward manner. To convert our algorithm to the asynchronous setting, one needs
to synchronize the reversal step. Synchronization is a heavy-handed approach and not desirable
in general.

To alleviate this problem, we believe, it is possible to get rid of the reversal step altogether.
The basic idea is to do away with the hard decisions to ``remove'' edges once a message from
a neighbor has been received. And instead to multiplicatively decrease the weight of such edges
for the next round. This approach would introduce a slight asymmetry in each edge's weight in both directions.
In order to analyze such an algorithm, it is needed to understand the behavior of
{\rfont RandomNeighbor} in the general asymmetric setting. In this setting, each vertex
uses its own distribution over outgoing links when choosing a communication partner at each step. We believe that understanding the asymmetric {\rfont RandomNeighbor} is an open problem of central importance.

\bibliographystyle{abbrv}
\bibliography{GossipConstraint}

\appendix

\section{Proof of Lemma~\ref{lemma:balcut}}
\label{app:balcut}

\BPF
Assume, towards a contradiction, that $\Phi(V- S) < \xi/3$. Then, there exists
a cut $(P,Q)$ of $V- S$ with $\varphi(P,Q)<\xi/3$ and specifically
\beq
\max\Bigg\{\frac{w(P,Q)}{\vol(P)},\frac{w(P,Q)}{\vol(Q)} \Bigg\}\le \frac{\xi}{3}
\eeq
Henceforth, let $Q$ be the smaller of the two, i.e. $\vol(Q)\le\vol(V- S)/2$.

We are going to show that $\varphi(S\cup Q,P) \le \xi$
and either $S\cup Q$ or $P$ should have been chosen instead of $S$.

Consider the case $\vol(S\cup Q) \le \vol(V)/2$. In this case,
\begin{gather*}
\varphi(S\cup Q,P)
	= \frac{w(S,P)+w(Q,P)}{\vol(S\cup Q)}
	= \frac{w(S,P)+w(Q,P)}{\vol(S) + \vol(Q)} \\
	\le \max\Bigg\{\frac{w(S,P)}{\vol(S)},\frac{w(Q,P)}{\vol(Q)}\Bigg\}
	\le \max\Bigg\{\frac{w(S,P)+w(S,Q)}{\vol(S)},\frac{w(Q,P)}{\vol(Q)}\Bigg\}
	\le \max\Big\{\xi, \xi/3\Big\} = \xi
\end{gather*}
This establishes a contradiction, because $\varphi(S\cup Q,P)\le\xi$
and $\vol(S)<\vol(S\cup Q)\le\vol(V)/2$.

Now let's consider the case $\vol(S\cup Q) > \vol(V)/2$.
First, we argue that $\vol(S\cup Q)$ cannot be too large.
We use that $\vol(Q)\le\frac{1}{2}\vol(V- S) = \frac{1}{2}(\vol(V)-\vol(S))$.
\beq
\vol(S\cup Q)
	= \vol(S)+\vol(Q)
	\le \vol(S)+\frac{\vol(V)-\vol(S)}{2}
	= \frac{\vol(V)+\vol(S)}{2}
	\le \frac{5}{8}\vol(V)
\eeq
Hence, $\vol(P) \ge \frac{3}{8}\vol(V)$. In addition, for the cut size, we have
\begin{align*}
w(S\cup Q, P)
	&= w(S,P)+w(Q,P) \\
	&\le \xi\vol(S) + \frac{\xi}{3}\vol(Q) \\
	&\le \xi\vol(S) + \frac{\xi}{3}\frac{\vol(V)-\vol(S)}{2} \\
	&\le \frac{5}{6}\xi\vol(S)+\frac{1}{6}\xi\vol(V) \\
	&=\frac{3}{8}\xi\vol(V)
\end{align*}
And now we can bound the cut conductance:
\beq
\varphi(S\cup Q,P)
	= \frac{w(S\cup Q,P)}{\vol(P)}
	\le \frac{\frac{3}{8}\xi\vol(V)}{\frac{3}{8}\vol(V)} = \xi
\eeq
This also establishes a contradiction because $\varphi(S\cup Q,P)\le\xi$ while
$\vol(S)\le\frac{1}{4}\vol(V)<\frac{3}{8}\vol(V)\le\vol(P)\le\frac{1}{2}\vol(V)$.
\EPF

\section{Proof of Theorem~\ref{thm:clstr}}
\label{app:clstr}

\BPF
The depth $K$ of the recursion is, by construction, at most $\log_{4/3}\vol(V)$ assuming
that the smallest non-zero weight is 1.
Let $\mathcal{R}_i\subseteq 2^V$ be a collection of the $U$-parameters of invocations
of \texttt{Cluster} at depth $0\le i\le K$ of the recursion.
(So, for example, $\mathcal{R}_0=\{V\}$.) For a set $U$ let $S(U)$ be the
small side of the cut produced by \texttt{Cluster($G$,$U$,$\xi$)}, or $\emptyset$ if no
eligible cut was found. We can then bound the total weight of cut edges as
\begin{align*}
\sum_{0\le i \le K} \sum_{U\in\mathcal{R}_i} w\big(S(U),U- S(U)\big)
	\le \sum_{0\le i \le K} \sum_{U\in\mathcal{R}_i} \xi\vol\big(S(U)\big)
	\le \sum_{0\le i \le K} \sum_{U\in\mathcal{R}_i} \frac{\xi}{2}\vol(U) \\
	\le \frac{\xi}{2}\sum_{0\le i \le K} \sum_{U\in\mathcal{R}_i} \vol(U)
	\le \frac{\xi}{2}\sum_{0\le i \le K} \vol(V)
	\le \frac{\xi \log_{4/3}\vol(V)}{2}\vol(V),
\end{align*}
Where we use the convention $w(\emptyset,S)=0$.
If we set $\xi = \frac{3\zeta}{\log_{4/3}\vol(V)}$, for some $0<\zeta<1$,
then Corollary~\ref{cor:balcut} establishes the theorem.
\EPF

\end{document}